\documentclass{article}
\usepackage{amssymb,amsmath,amsthm,mathrsfs,empheq}
\usepackage{graphicx,color} 
\usepackage{fullpage}
\newcommand{\NP}{\ensuremath{\mathrm{NP}}}
\newcommand{\W}{\ensuremath{\mathrm{W}}}
\newcommand{\FPT}{\ensuremath{\mathrm{FPT}}}

\newcommand{\tw}{{\mathbf{tw}}}
\newcommand{\pw}{{\mathbf{pw}}}
\newcommand{\vc}{{\mathbf{vc}}}

\newcommand{\dist}{\mathrm{dist}}
\newcommand{\diam}{\mathrm{diam}}

\newtheorem{theorem}{Theorem}
\newtheorem{lemma}{Lemma}

\newtheorem{proposition}{Proposition}

\begin{document}
\title{Finding vertex-surjective graph 
homomorphisms\thanks{A preliminary version of this paper appeared as an extended abstract in the proceedings of
CSR 2012. The work was supported by 
 EPSRC (EP/G043434/1 and EP/G020604/1) and the Royal Society (JP090172).
The 2nd author was  supported by Charles University as GAUK~95710.}
}

\author{
%Petr A. Golovach\inst{1}, Bernard Lidick\'y\inst{2}, \\Barnaby Martin\inst{1}, and~Dani\"el Paulusma\inst{1}}
Petr A. Golovach\thanks{School of Engineering and Computing Sciences, Durham University, Durham, U.K. \texttt{petr.golovach@durham.ac.uk}} ,
Bernard Lidick\'y\thanks{Department of Mathematics, University of Illinois, Urbana, USA and Department of Applied Mathematics, Charles University in Prague, Czech Republic \texttt{bernard@kam.mff.cuni.cz}}, \\
Barnaby Martin\thanks{School of Engineering and Computing Sciences, Durham University, Durham, U.K. \texttt{barnaby.martin@durham.ac.uk}}, 
and~Dani\"el Paulusma\thanks{School of Engineering and Computing Sciences, Durham University, Durham, U.K. \texttt{daniel.paulusma@durham.ac.uk}}
}
%
%\institute{
%School of Engineering and Computing Sciences, Durham University\\ 
%Science Laboratories, South Road, Durham DH1 3LE, UK,\\
%\texttt{\{petr.golovach,barnaby.martin,daniel.paulusma\}@durham.ac.uk}
%\and
% Charles University, Faculty of Mathematics and Physics,\\
%Malostransk\'e n\'am.~2/25, 118~00, Prague, Czech Republic\\
%Department of Mathematics, University of Illinois\\
%Urbana, IL 61801, USA, \email{bernard@kam.mff.cuni.cz}
%}

\maketitle
\vspace{-0.45cm}
\begin{abstract}
 The {\sc Surjective Homomorphism} problem  is to test whether a given graph
$G$ called the guest graph allows a vertex-surjective homomorphism to some other given graph $H$ called the host graph. The bijective and injective homomorphism problems can be formulated in terms of spanning subgraphs and subgraphs, and as such their computational complexity has been extensively studied.
What about the surjective variant?
Because this problem is   \NP-complete in general, we restrict the guest and the host graph to belong to graph classes ${\cal G}$ and ${\cal H}$, respectively. We  determine to what extent a certain choice of ${\cal G}$ and ${\cal H}$ influences 
its computational complexity.
We observe 
that the problem is polynomial-time solvable if ${\cal H}$ is the class of paths, whereas it is \NP-complete if ${\cal G}$ is the class of paths. Moreover, we show that the problem
is even \NP-complete on many other elementary graph classes, namely
linear forests, unions of complete graphs, cographs,
proper interval graphs, split graphs and trees of pathwidth at most 2. In contrast, we prove that the problem is fixed-parameter tractable in $k$ if ${\cal G}$ is the class of trees and ${\cal H}$ is the class of trees with at most $k$ leaves, or if ${\cal G}$ and ${\cal H}$ are equal to the class of graphs with vertex cover number at most~$k$.
\end{abstract}

\section{Introduction}\label{s-intro}
We  consider undirected finite graphs that are {\it simple}, i.e., have no loops and no multiple edges.
A graph is denoted $G=(V_G, E_G)$, where $V_G$ is the set of vertices and $E_G$ is the set of edges.
A {\it homomorphism} from a graph $G$
to a graph $H$  is a mapping $f: V_G \to V_H$ that maps adjacent vertices of $G$ to adjacent vertices of $H$,  i.e., $f(u)f(v)\in E_H$ whenever $uv\in E_G$.
Graph homomorphisms are widely studied within the areas of graph theory and algorithms; for a survey we refer to  the monograph of Hell and Ne\v{s}et\v{r}il~\cite{HN04}.
The {\sc Homomorphism} problem is to test whether there exists a homomorphism from a graph $G$ called the {\it guest graph} to a graph $H$ called the {\it host graph}. If $H$ is restricted to be in the class of complete graphs (graphs with all possible edges), then this problem is equivalent to the {\sc Coloring} problem. The latter problem is to test whether a graph $G$ allows a {\it $k$-coloring} for some given $k$, i.e., a mapping $c:V_G\to \{1,\ldots,k\}$, such that $c(u)\neq c(v)$ whenever $uv\in E_G$. This is a classical \NP-complete problem~\cite{GareyJ79}. Hence, the {\sc Homomorphism} problem is \NP-complete in general, and it is natural to restrict the input graphs to belong to some special graph classes. 

We let ${\cal G}$ denote the class of guest graphs and ${\cal H}$ the class of host graphs that are under consideration, and denote the corresponding decision problem by  {\sc $({\cal G},{\cal H})$-Homomorphism}. If $G$ or $H$ is the class of all graphs, then
we use the notation ``$-$'' to indicate this. If ${\cal G}=\{G\}$ or ${\cal H}=\{H\}$, we write 
$G$ and $H$ instead of ${\cal G}$ and ${\cal H}$, respectively,
The Hell-Ne\v{s}et\v{r}il dichotomy theorem~\cite{HN90} states that $(-,H)$-{\sc Homomorphism} is solvable in polynomial time if $H$ is bipartite, and \NP-complete otherwise.
In the context of graph homomorphisms, a  graph $F$ is called a {\it core} if there exists no homomorphism from $F$ to any proper subgraph of $F$. 
Dalmau et al.~\cite{DalmauKV02} proved that the $({\cal G},-)$-{\sc Homomorphism} problem can be solved in polynomial time if all cores of the graphs in ${\cal G}$ have bounded treewidth.
Moreover, Grohe~\cite{Gr07} showed that under the assumption $\FPT\neq\W[1]$, the problem can be solved in polynomial time if and only if this condition holds.

As a homomorphism $f$ from a graph $G$ to a graph $H$ is a (vertex) mapping, we may add further restrictions, such as requiring it to be {\it bijective}, {\it injective}, or {\it surjective} i.e., for each $x\in V_H$ there exists exactly one, at most one, or at least one vertex $u\in V_G$ with $f(u)=x$,  respectively.
The decision problems corresponding to the first and second variant are known as the
{\sc Spanning Subgraph Isomorphism} and {\sc Subgraph Isomorphism} problem, respectively. As such, 
these two variants have been well studied in the literature.
For example, the bijective variant contains the problem that is to test whether a graph contains a Hamiltonian cycle as a special case.
In our paper, we research the third variant, which
leads to the following decision problem:

\medskip
\noindent
{\sc Surjective Homomorphism}\\
{\it Instance:} two graphs $G$ and $H$.\\
{\it Question:} does there exist a surjective homomorphism from $G$ to $H$?

\medskip
\noindent
If the guest $G$ is restricted to a graph class ${\cal G}$ and  the host $H$ to a graph class ${\cal H}$, then we denote this problem by {\sc Surjective $({\cal G},{\cal H})$-Homomorphism}. 
Fixing the host side to a single graph $H$ yields the  {\sc Surjective $(-,H)$-Homomorphism} problem.
This problem is \NP-complete already when $H$ is nonbipartite. This follows from a simple reduction from the corresponding $(-,H)$-{\sc Homomorphism} problem, which is \NP-complete due to the Hell-Ne\v{s}et\v{r}il dichotomy theorem~\cite{HN90};
we replace an instance graph $G$ of the latter problem by the disjoint union $G+H$ of $G$ and $H$, and observe that $G$ allows an homomorphism to $H$ if and only if $G+H$ allows a surjective homomorphism to $H$. 
For bipartite host graphs $H$, the complexity classification of {\sc Surjective $(-,H)$-Homomorphism} is still open, although some partial results are known. 
For instance, the problem can be solved in polynomial time whenever $H$ is a tree. This follows from a more general classification that also includes trees in which the vertices may have self-loops~\cite{GPS11}.
On the other hand,  there  exist cases of bipartite host graphs $H$ for which the problem is \NP-complete, e.g., when $H$ is the graph obtained from a 6-vertex cycle with  one distinct path of length 3 added to each of its six vertices~\cite{BKM11}.
Recently,  the {\sc Surjective $(-,H)$-Homomorphism} problem has been shown to be \NP-complete when $H$ is a 4-vertex cycle with a self-loop in every vertex~\cite{MP11}.  Note that in our paper we only consider simple graphs. For a survey on 
 the {\sc Surjective $(-,H)$-Homomorphism} problem from a constraint satisfaction point of view we refer to the paper of Bodirsky, K\'{a}ra and Martin~\cite{BKM11}.
Below we discuss some other concepts that are closely related to 
surjective homomorphisms.

A homomorphism $f$ from a graph $G$ to a graph $H$ is {\it locally
surjective} if $f$ becomes surjective when restricted to the neighborhood of every vertex $u$ of $G$, i.e., $f(N_G(u))=N_H(f(u))$.  The corresponding decision is called the {\sc Role Assignment} problem which has been classified for any fixed host $H$~\cite{FP05}.
Any locally surjective homomorphism is surjective if the host graph is connected but the reverse implication is not true in general.  For more on locally surjective homomorphisms and the locally injective and bijective variants, we refer to the survey of  Fiala and Kratochv\'il~\cite{FK08}.

Let $H$ be an induced subgraph of a graph $G$.
Then a  homomorphism $f$ from a graph $G$ to $H$
is a {\it retraction} from $G$ to $H$ if $f(h)=h$ for all $h\in V_H$. 
In that case we say that $G$ {\it retracts to} $H$. By definition, a retraction from $G$ to
$H$ is a surjective homomorphism from $G$ to $H$. Retractions are well studied; see e.g. the recent complexity classification of Feder et al.~\cite{FHJKN10} for the corresponding decision problem when $H$ is a fixed pseudoforest.
In particular, polynomial-time algorithms for retractions have been proven to be a useful subroutine for obtaining polynomial-time algorithms for the {\sc Surjective $(-,H)$-Homomorphism} problem~\cite{GPS11}.

We emphasize that  a surjective homomorphism is  {\it vertex-surjective} as opposed to the stronger condition of being edge-surjective. A homomorphism from a graph $G$ to a graph $H$ is called {\it edge-surjective} or 
a {\it compaction} if for any edge $xy\in E_H$ there exists an edge $uv\in E_G$ with $f(u)=x$ and $f(v)=y$. If $f$ is a compaction from $G$ to $H$, we also say that $G$ {\it compacts} to $H$.
The {\sc Compaction} problem is to test whether a graph $G$ compacts to a 
graph~$H$. Vikas~\cite{Vi02,Vi04,Vi05} determined the computational complexity of 
$(-,H)$-{\sc Compaction} 
 for several classes of fixed host graphs $H$. 
Very recently, Vikas~\cite{Vi11} 
considered $(-,H)$-{\sc Compaction} for
guest graphs belonging to some restricted graph class.

\medskip
\noindent
{\bf Our Results.}
We study the {\sc Surjective $({\cal G},{\cal H})$-Homomorphism} problem for several graph classes ${\cal G}$ and ${\cal H}$.  
We observe that this problem is polynomial-time solvable when the host graph is a path, whereas
it becomes \NP-complete if we restrict the guests to be paths instead of the hosts. 
We also show that the problem is \NP-complete when both 
${\cal G}$ and ${\cal H}$ are restricted to trees of pathwidth at most 2, and when both ${\cal G}$ and ${\cal H}$ are linear forests. These results are in contrast to the aforementioned polynomial-time result of 
Dalmau et al.~\cite{DalmauKV02} on $({\cal G},-)$-{\sc Homomorphism} for graph classes ${\cal G}$ that consists of graphs, the cores of which have bounded treewidth. They are also in contrast to the aforementioned polynomial-time result on {\sc Surjective $(-,H)$-Homomorphism} when $H$ is any fixed tree~\cite{GPS11}. 

Due to the hardness for graphs of bounded treewidth, it is natural to consider other width parameters such as the clique-width of a graph. For this purpose we first consider the class
of complete graphs that are exactly those graphs that have clique-width 1. We observe that the
{\sc Surjective $({\cal G},{\cal H})$-Homomorphism} can be solved in polynomial time when
${\cal G}$ is the class of complete graphs, whereas the problem becomes \NP-complete when
we let ${\cal G}$ and ${\cal H}$ consist of the unions of complete graphs.  We then focus on graphs that have clique-width at most two. This graph class is equal to the class of cographs~\cite{CourcelleO00}.
There exist only a few natural problems that are difficult on cographs. We prove that {\sc Surjective $({\cal G},{\cal H})$-Homomorphism}, where ${\cal G}$ and ${\cal H}$ are equal to the class of connected cographs, is one of these.
We also consider proper interval graphs.  This graph class has unbounded tree-width and contains the classes of complete graphs and paths. Because they are ``path-like'', often problems that are difficult for general graphs are tractable for proper interval graphs. In an attempt  to generalize  our polynomial-time result for {\sc Surjective $({\cal G},{\cal H})$-Homomorphism}  when
${\cal G}$ is the class of complete graphs, or when  ${\cal H}$ is the class of paths, we consider connected proper interval graphs. It turns out that {\sc Surjective $({\cal G},{\cal H})$-Homomorphism} is \NP-complete even  when ${\cal G}$ and ${\cal H}$ consist of these graphs.
Our last hardness result shows that the problem is also \NP-complete when ${\cal G}$ and ${\cal H}$ are equal to the class of split graphs. All hardness results can be found in Section~\ref{sec:NPc}.

To complement our hardness results, we show in Section~\ref{sec:FPT} that  
{\sc Surjective $({\cal G},{\cal H})$-Homomorphism} 
is fixed-parameter tractable in $k$, when ${\cal G}$ is the class of trees and ${\cal H}$ is the class of trees with 
at most $k$ leaves, and also when ${\cal G}$ and ${\cal H}$ consist of graphs with vertex cover number at most  $k$. The latter result adds further evidence 
that decision problems difficult for graphs of bounded treewidth may well be tractable if the vertex cover number is bounded; also see e.g.~\cite{AdigaCS10,EncisoFGKRS10,FellowsLMRS08,FialaGK11}.  Moreover, the vertices  
of such graphs can be partitioned into two sets, one of them has size bounded by the vertex cover number and the other one is an independent set.
As such, they resemble split graphs with bounded clique number.
We refer to Table~\ref{tabl:compl} for a summary of our results.
In this table, $\pw$ and $\vc$ denote the pathwidth and the vertex cover number of a graph, respectively. 
In Section~\ref{sec:defs} we explain these notions and the complexity class \FPT . There, we also give the definitions of all the aforementioned graph classes.

\begin{table}
\begin{center}
\vskip-5mm
\begin{tabular}{|l|l|l|l|}
\hline
\hspace*{1.2cm}$\cal G$ & \hspace*{1.2cm}$\cal H$ &Complexity &\\
\hline
complete graphs  &all graphs & polynomial time &Proposition~\ref{prop:H-complete} (i)\\
\hline
all graphs & paths & polynomial time &Proposition~\ref{prop:H-complete} (ii)\\
\hline
paths & all graphs & \NP-complete &Theorem~\ref{t-all} (i)\\
\hline
linear forests & linear forests & \NP-complete  &Theorem~\ref{t-all} (ii)\\
\hline
unions of complete graphs & unions of complete graphs & \NP-complete  &Theorem~\ref{t-all} (iii)\\
\hline
connected cographs & connected cographs & \NP-complete &Theorem~\ref{t-all} (iv)\\
\hline
trees of $\pw\leq2$ & trees of $\pw\leq2$ & \NP-complete &Theorem~\ref{t-all} (v)\\
\hline
split graphs  & split graphs & \NP-complete &Theorem~\ref{t-all} (vi)\\
\hline
connected proper  & connected proper & \NP-complete &Theorem~\ref{t-all} (vii)\\
interval graphs & interval graphs & &\\
\hline
trees & trees with $k$ leaves & \FPT~in $k$  &Theorem~\ref{thm:trees-FPT}\\
\hline
graphs of $\vc\leq k$ & graphs of $\vc\leq k$ & \FPT~in $k$ &Theorem~\ref{thm:vcn-FPT}\\
\hline
\end{tabular}
\vskip2mm
\caption{Complexity of {\sc $({\cal G},{\cal H})$-Surjective Homomorphism.}}\label{tabl:compl}
\end{center}
\vskip-10mm
\end{table}

\section{Definitions and Preliminaries}\label{sec:defs}
Let $G$ be a graph.
 The \emph{open neighborhood} of a vertex $u\in V_G$ is defined as $N_G(u) = \{v\; |\; uv\in E_G\}$, and
its \emph{closed neighborhood} is defined as $N_G[u] = N_G(u) \cup \{u\}$. The degree of a vertex
$u\in V_G$ is denoted $d_G(u)=|N_G(u)|$.
The \emph{distance} $\dist_G(u,v)$ between a pair of vertices $u$ and $v$ of
$G$ is the number of edges of a shortest path between them. 
The \emph{distance} between a vertex $u$ and a set of vertices $S\subseteq V_G$ is
$\dist_G(u,S)=\min\{\dist_G(u,v)|v\in S\}$. 
We may omit subscripts if this does not create any confusion.
The \emph{diameter} of $G$ is defined as $\diam(G)=\max\{\dist_G(u,v)|u,v\in V_G\}$.
Let $S\subseteq V_G$. Then the graph $G-S$ is the graph obtained from $G$ by removing all vertices in $S$. If $S=\{u\}$, we also write $G-u$. The subgraph of $G$ that is {\it induced} by $S$ has vertex set $S$ and edges $uv$ if and only if $uv\in E_G$. We denote this subgraph by $G[S]$.

A graph is an {\it interval graph} if intervals of the real line can be
associated with its vertices in such a way that two vertices are adjacent if and only
if their corresponding intervals overlap.    An interval graph is {\it proper} if it has an interval representation, in which no interval is properly contained in any other interval.
The disjoint union of two graphs $G$ and $H$ is denoted $G+H$, and
the disjoint union of $r$ copies of $G$ is denoted $rG$.
A {\it linear forest} is the disjoint union of a collection of paths.
We denote the path on $n$ vertices by $P_n$.
A graph is a \emph{cograph} if it does not contain $P_4$ as an induced subgraph.
A {\it clique} is the vertex set of a complete graph. A vertex set is {\it independent} if its vertices are mutually non-adjacent.
A graph is a \emph{split graph} if its vertex set can be partitioned into a clique and an independent set.

A \emph{tree decomposition} of a graph $G$ is a pair $({\cal X},T)$ where $T$
is a tree and ${\cal X}=\{X_{i} \mid i\in V_T\}$ is a collection of subsets (called {\em bags})
of $V_G$ such that the following three conditions are satisfied:
\begin{enumerate}
\item $\bigcup_{i \in V_T} X_{i} = V_G$;
\item for each edge $xy \in E_G$, the vertices $x,y$ are in a bag $X_i$ for some  $i\in V_T$;
\item for each $x\in V_G$, the set $\{ i \mid x \in X_{i} \}$ induces a connected subtree of $T$.
\end{enumerate}
The \emph{width} of tree decomposition $({\cal X},T)$ is $\max_{i \in V_T}\,\{|X_{i}| - 1\}$. The \emph{treewidth} of a graph $G$, denoted $\tw(G)$, is the minimum width over all tree decompositions of $G$.  If in these two definitions we restrict the tree $T$ to be a path, 
then we obtain the notions of \emph{path decomposition} and {\em
pathwidth} of $G$ denoted $\pw(G)$.

For a graph $G$, a set $S\subseteq V_G$ is a \emph{vertex cover}
of $G$, if every edge of $G$ has at least one of its two endvertices in $S$.
Let $\vc(G)$ denote the \emph{vertex cover number}, i.e., the minimum size of a vertex cover of $G$.

We use the following well-known notion in parameterized complexity, where one considers the problem input as a pair $(I,k)$, where $I$ is the main part and $k$ 
the parameter; also see the text book of 
Flum and Grohe~\cite{FlumG06}. 
 A problem is \emph{fixed parameter tractable} if an instance $(I,k)$ can be solved
in time $O(f(k) n^c)$, where $f$ denotes a computable function, $n$ denotes the size of $I$, and $c$ is a constant independent of $k$.  
The class \FPT\  is the class of all fixed-parameter tractable decision problems.

We finish this section by giving the polynomial-time results from Table~\ref{tabl:compl}.  
The proof of statement (ii) of Proposition~\ref{prop:H-complete} is similar to the corresponding proof for the edge-surjective variant shown by Vikas~\cite{Vi11}. 

\begin{proposition}\label{prop:H-complete}
The {\sc Surjective $({\cal G},{\cal H})$-Homomorphism} problem can be solved in polynomial time in the following two cases:
\begin{itemize}
\item[(i)] 
${\cal G}$ is the class of complete graphs and ${\cal H}$ is the class of all graphs;
\item[(ii)] ${\cal G}$ is the class of all graphs and ${\cal H}$ is the class of paths.
\end{itemize}
\end{proposition}

\begin{proof}
We first prove (i).  Let $G$ be a complete graph and $H$ be an arbitrary graph.
We claim that there exists a surjective homomorphism from $G$ to $H$ if and only if $H$ is a complete graph with the same number of vertices as $G$. Because this condition can be checked in polynomial time, showing this is sufficient to prove (i).

First suppose that $H$ is a complete graph with the same number of vertices as $G$. 
Then the identity mapping is 
a surjective homomorphism from $G$ to $H$. 

Now suppose that $f$ is a surjective homomorphism from $G$ to $H$.
Because $G$ is a complete graph and $f$ is a homomorphism, there are no two distinct vertices $u$ and $v$ with 
$f(u)=f(v)$. Because $f$ is surjective, this means that $|V_G|=|V_H|$. 
Let $x$ and $y$ be two distinct vertices in $H$.
Because $f$ is surjective, there exist two vertices $u$ and $v$ in $G$ with $f(u)=x$ and $f(v)=y$.
Because $G$ is a complete graph, $u$ and $v$ are adjacent. Then, because $f$ is a homomorphism, $x$ and $y$ 
must be adjacent. Hence, $H$ is a complete graph. 
This completes the proof of (i).

\medskip
We now prove (ii).
Suppose that we are given a guest graph  $G$ with $k$ connected components $G_1,\ldots,G_k$ for some $k\geq 1$, and
a host path $P_\ell$ for some $\ell\geq 1$. 
If $\ell=1$, then there exists a surjective homomorphism from $G$ to $P_\ell$ if and only if
each $G_i$ consists of one vertex. Assume that $\ell\geq 2$.
We claim that there exists a surjective homomorphism from $G$ to $P_\ell$ if and only if
a) $G$ is bipartite and b) $\sum_{i=1}^k\diam(G_i)+k\geq \ell$.
Because conditions a) and b) can be checked in polynomial time, showing this is sufficient to prove (ii).

First suppose that $f$ is a surjective homomorphism from $G$ to $P_{\ell}$.
Because $P_{\ell}$ is a bipartite graph, $G$ is bipartite as well, and a) holds. 
For $i=1,\ldots,k$, we let  $P^i$ denote the subgraph of $P_{\ell}$ induced by $f(V_{G_i})$. 
Because each $G_i$ is connected and $f$ is a homomorphism, each $P^i$ is connected, and hence, forms a subpath of $P_\ell$.
Because $f$ is a homomorphism, $\diam(G_i)\geq |V_{P^i}|-1$ for $i=1,\ldots,k$. We use this inequality and the surjectivity of $f$ 
to obtain
$$\ell\leq \sum_{i=1}^k|V_{P^i}|
\leq \sum_{i=1}^k(\diam(G_i)+1)=\sum_{i=1}^k\diam(G_i)+k.$$

Now suppose that $G$ is bipartite and that  $\sum_{i=1}^k\diam(G_i)+k\geq \ell$.
Let  $F=G_i$ be an arbitrary connected component of $G$. We  first prove that for all
$\min\{1,\diam(F)\}+1\leq s\leq \diam(F)+1$, there is a surjective homomorphism $h$ from $F$ to $P_s$. Clearly, this holds if $\diam(F)=0$. 
Let $\diam(F)\geq 1$. Then $s\geq 2$. Let $P=v_1v_2 \cdots v_s$. 
Let $u$ be a vertex of $F$ such that $F$ has a vertex at distance $\diam(F)$ from $u$. 
We consider a mapping $h\colon V_{F}\rightarrow \{v_1,\ldots,v_k\}$ such that $h(x)=v_i$, where
$$i=
\begin{cases}
\dist_F(u,x)+1 & \text{if $\dist_F(u,x)\leq s-2$,}\\
s & \text{if $\dist_F(u,x)\geq s-1\; \mbox{and}\; (\dist_F(u,x)-s+1)~{\rm mod}~ 2=0,$}\\
s-1& \text{if $\dist_F(u,x)\geq s-1\; \mbox{and}\; (\dist_F(u,x)-s+1)~{\rm mod}~2=1.$}
\end{cases} 
$$
Because $G$ is bipartite, $F$ is bipartite. Then $h$ is a homomorphism, and because $\diam(G)\geq s-1$, $h$ is surjective.

Because $\sum_{i=1}^k\diam(G_i)+k\geq \ell\geq 2$,  we can cover $P_{\ell}$ by subpaths $P^1,\ldots,P^k$ (i.e., $\cup_{i=1}^kV_{P^i}=V_{P_{\ell}}$) in such a way that 
for all $1\leq i\leq k$, we have that $\min\{1,\diam(G_i)\}+1\leq |V_{P^i}|\leq \diam(G_i)+1$. 
It remains to recall that we have a surjective homomorphism from each $G_i$ to $P^i$, and claim b) follows.
This completes the proof of (ii), and hence, we have shown Proposition~\ref{prop:H-complete}.\qed
\end{proof}

\section{Hard Cases}\label{sec:NPc}
In contrast to case (ii) of Proposition~\ref{prop:H-complete}, where the host graphs are assumed to be paths, our problem becomes difficult when the guest graphs are restricted to paths.
Our next theorem shows this and the other hardness results of Table~\ref{tabl:compl}.

\begin{theorem}\label{t-all}
The {\sc Surjective $({\cal G},{\cal H})$-Homomorphism} problem is \NP-com\-plete in the following
six cases:
\begin{enumerate}
\item [(i)] ${\cal G}$ is the class of paths and ${\cal H}$ is the class of all graphs;
\item [(ii)]  ${\cal G}={\cal H}$ is the class of linear forests;
\item [(iii)] ${\cal G}={\cal H}$ is the class of disjoint unions of complete graphs;
\item [(iv)] ${\cal G}={\cal H}$ is the class of connected cographs;
\item [(v)] ${\cal G}={\cal H}$ is the class of trees of pathwidth at most two;
\item [(vi)] ${\cal G}={\cal H}$ is the class of split graphs;
\item [(vii)] ${\cal G}={\cal H}$ is the class of connected proper interval graphs.
\end{enumerate}
\end{theorem}

\begin{proof}
We first prove (i). We reduce from the well-known problem {\sc Hamiltonian Path}, which is  \NP-complete~\cite{GareyJ79}. An $n$-vertex graph $H$ has a Hamiltonian path
if and only if there exists a surjective homomorphism from $P_n$ to $H$. This proves~(i).

For showing (ii)-(vii) we need some extra terminology.
We say that a multiset $A=\{a_1,\ldots,a_n\}$ of integers is {\it $(m,B)$-positive} if $n=3m$, $\sum_{i=1}^na_i=mB$ and $a_i>0$ for $i=1,\ldots,n$.
A \emph{$3$-partition} of a multiset $A=\{a_1,\ldots,a_n\}$ that is $(m,B)$-positive for some integers $m,B$ is a partition $S_1,S_2,\ldots, S_m$ of $A$ such  that for $1\leq j\leq m$,  
$|S_j|=3$ and $\sum_{a_i\in S_j}a_i=B$.
This leads to the  problem:

\medskip
\noindent
 {\sc $3$-Partition}\\
 {\it Instance:} an $(m,B)$-positive multiset $A=\{a_1,\ldots,a_{n}\}$  for some integers $m,B$;\\
 {\it Question:} does $A$ have a 3-partition?
 
 \medskip
 \noindent
The {\sc $3$-Partition} problem is known to be \NP-complete~\cite{GareyJ79} in the strong sense, i.e., it remains hard even if all integers in the input are  encoded in unary.  This enables us to reduce from this problem  in order to show \NP-completeness in the cases (ii)-(vii). In each of these six cases we assume that $A=\{a_1,\ldots,a_{n}\}$ is a  $(m,B)$-positive multiset  for some integers $m,B$.  We now
prove (ii)-(vii).
 
\medskip
\noindent
{\bf (ii)}  For $i=1,\ldots,n$, let $p_i=a_i+B$, and let $q=4B$.
Let $G$ be the linear forest $G_1+\cdots + G_n$, where $G_i$ is isomorphic to $P_{p_i}$
for $i=1,\ldots,n$.
Let $H$ be the linear forest $H_1+\cdots + H_m=mP_q$.
The forests $G$ and $H$ are displayed in Figure~\ref{f-linforest}. %BL
We claim that $A$ has a $3$-partition if and only if there exists a surjective homomorphism from 
$G$ to $H$. 

\begin{figure} %BL
\begin{centering}
\includegraphics[scale=1]{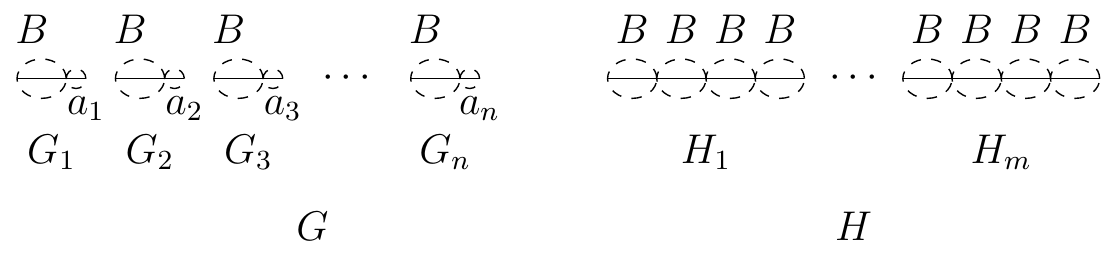}
\caption{The linear forests $G$ and $H$ constructed in the proof of (ii).}
\label{f-linforest}
\end{centering}
\end{figure}

Suppose that $S_1,\ldots,S_m$ is a $3$-partition of $A$. For each $1\leq j \leq m$, we consider the connected components $G_{i_1},G_{i_2},G_{i_3}$ of $G$ such that $S_j=\{a_{i_1},a_{i_2},a_{i_3}\}$. We map the vertices of $G_{i_1}$ to the first $p_{i_1}$ vertices of $H_j$ according to the path order, and similarly the vertices of $G_{i_2}$ to the next $p_{i_2}$ vertices of $H_j$, and the vertices of $G_{i_3}$ to the last $p_{i_3}$ vertices of $H_j$. Because $p_{i_1}+p_{i_2}+p_{i_3}=a_{i_1}+a_{i_2}+a_{i_3}+3B=4B=q$ for $i=1,\ldots,n$, we obtain a surjective homomorphism from $G$ to $H$ in this way.
 
Now suppose that $f$ is a surjective homomorphism from $G$ to $H$. 
We observe that $|V_G|=|V_H|=4mB$. Hence, $f$ is also injective. Because $f$ is a homomorphism, 
$f$ must map all vertices of each connected component of $G$ to the same connected component of $H$. 
Let $1\leq j\leq m$, and let
$G_{i_1},\ldots,G_{i_s}$ be the connected components of $G$ that are  mapped to $H_j$.  
Because $|V_{H_j}|=4B$ and every connected component of $G$ contains at least $B+1$ vertices, we  find that $s\leq 3$.
Because $G$ has $3m$ connected components, we then find that $s=3$. Because $f$ is injective, $a_{i_1}+a_{i_2}+a_{i_3}+3B=p_{i_1}+p_{i_2}+p_{i_3}=q=4B$. 
Hence, $a_{i_1}+a_{i_2}+a_{i_3}=B$ and we let $S_j=\{a_{i_1},a_{i_2},a_{i_3}\}$. 
This means that we obtain the partition $S_1,\ldots,S_m$ of $A$ that is a $3$-partition. 
This completes the proof of (ii).

\medskip
\noindent
{\bf (iii)} We use all arguments from the proof of (ii) after replacing each path in $G$ and $H$ by a clique of the same size.

\medskip
\noindent
{\bf (iv)} In the graphs $G$ and $H$ from the proof of (ii) we replace each path by a clique of the same size. We also add a vertex $v$ in $G$ adjacent  to all other vertices of $G$, and a vertex $x$ in $H$ adjacent to all other vertices of $H$.
The resulting graphs are connected cographs.  We observe that every homomorphism maps $v$ to~$x$.
To finish the proof we use the same arguments as the ones used to prove~(ii).

\begin{figure}
\begin{centering}
\includegraphics[scale=0.95]{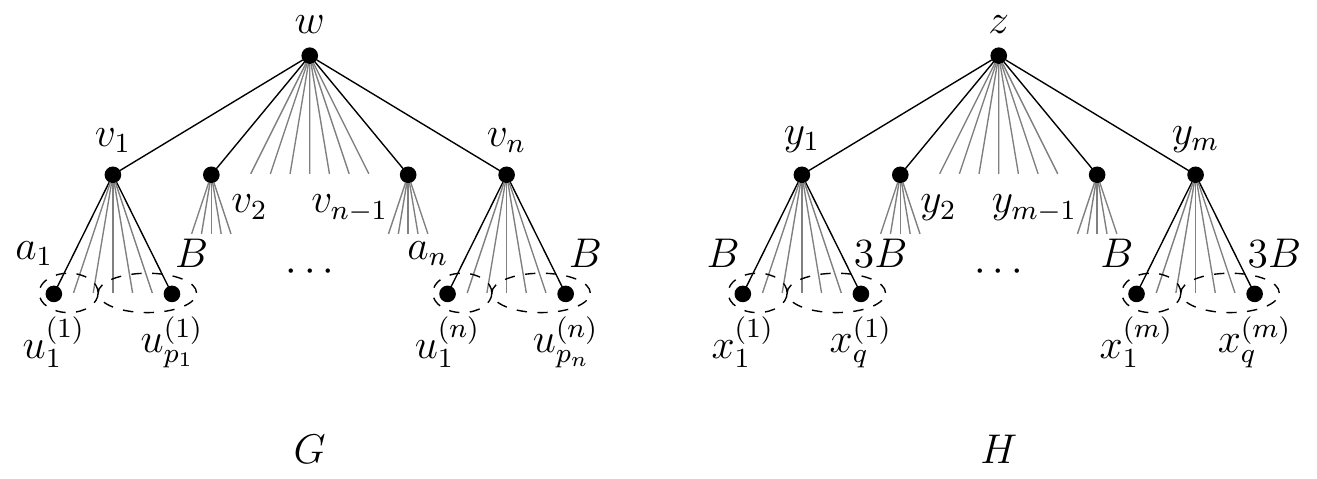}
\caption{The trees $G$ and $H$ constructed in the proof of (v).}
\label{f-ghdisplay}
\end{centering}
\end{figure}
\noindent
{\bf (v)} For $i=1,\ldots,n$, let $p_i=a_i+B$, and let $q=4B$.
We construct two trees $G$ and $H$. We first construct $G$:
\begin{itemize} 
\item[$\bullet$] for  $i=1,\ldots,n$, introduce $p_i$ vertices $u_1^{(i)},\ldots,u_{p_i}^{(i)}$ and a vertex $v_i$ adjacent to $u_1^{(i)},\ldots,u_{p_i}^{(i)}$;
\item[$\bullet$] add a new vertex $w$ and make it adjacent to $v_1,\ldots,v_n$.
\end{itemize}
We now construct $H$:
\begin{itemize} 
\item[$\bullet$] for $j=1,\ldots,m$, introduce $q$ vertices $x_1^{(j)},\ldots,x_{q}^{(j)}$ and a vertex $y_j$ adjacent to $x_1^{(j)},\ldots,x_{q}^{(j)}$;
\item[$\bullet$] add a new vertex $z$ and make it  adjacent to $y_1,\ldots,y_m$.
\end{itemize}
The trees $G$ and $H$ are displayed in Figure~\ref{f-ghdisplay}.
For $G$ we take the path decomposition with bags $\{u_h^{(i)},v_i,w\}$ to find that $\pw(G)\leq 2$.
Similarly, we find that $\pw(H)\leq 2$.
We claim that $A$ has a $3$-partition if and only if there is a surjective homomorphism from $G$ to $H$.

First suppose that $S_1,\ldots,S_m$ is a $3$-partition of $A$. We define $f$ as follows. 
We set $f(w)=z$. Then
for $j=1,\ldots,m$, we consider the set $S_j=\{a_{i_1},a_{i_2},a_{i_3}\}$. 
We let $f$ map the vertices $v_{i_1},v_{i_2},v_{i_3}$ to $y_j$.
Then we let $f$ map the vertices $u_1^{(i_1)},\ldots,u_{p_{i_1}}^{(i_1)}$ 
 consecutively to the first $p_{i_1}$ vertices of the set
$\{x_1^{(j)},\ldots,x_{q}^{(j)}\}$, the vertices 
$u_1^{(i_2)},\ldots,u_{p_{i_2}}^{(i_2)}$ 
to the next $p_{i_2}$ vertices of 
this set, and finally, the vertices $u_1^{(i_3)},\ldots,u_{p_{i_3}}^{(i_3)}$ to 
 the last $p_{i_3}$ vertices of the set. Because $p_{i_1}+p_{i_2}+p_{i_3}=a_{i_1}+a_{i_2}+a_{i_3}+3B=4B=q$, we find that $f$
 is a surjective homomorphism from $G$ to $H$.
  
Now suppose that $f$ is a surjective homomorphism from $G$ to $H$. We observe that $f(w)=z$, because all vertices of $G$ must be mapped at distance at most two from $f(w)$. 
Consequently, $f$ maps every $v$-vertex to a $y$-vertex, and every $u$-vertex to an $x$-vertex.
The number of $u$-vertices is $p_1+\ldots+p_n=a_1+\ldots+a_n+nB=4mB$, which is equal to the number
of $x$-vertices. 
Hence $f$ maps the $u$-vertices  
 bijectively to the $x$-vertices.
 Moreover, if $f(v_i)=y_j$, then $f$ maps the vertices  $u_1^{(i)},\ldots,u_{p_i}^{(i)}$ to the vertices from the set $\{x_1^{(j)},\ldots,x_{q}^{(j)}\}$.
For $j=1,\ldots,m$, 
let $v_{i_1},\ldots,v_{i_s}$ be the vertices mapped to $y_j$.  
Because $p_i>B$ for all $1\leq i\leq n$, we find that $s\leq 3$. Then, because $n=3m$, we conclude that $s=3$. Because $f$ maps 
bijectively 
$\{u_1^{(i_1)},\ldots,u_{p_{i_1}}^{(i_1)}\}\cup \{u_1^{(i_2)},\ldots,u_{p_{i_2}}^{(i_2)}\}
\cup \{u_1^{(i_3)},\ldots,u_{p_{i_3}}^{(i_3)}\}$ to
$\{x_1^{(j)},\ldots,x_{q}^{(j)}\}$, we find that $a_{i_1}+a_{i_2}+a_{i_3}+3B=p_{i_1}+p_{i_2}+p_{i_3}=q=4B$, and consequently, $a_{i_1}+a_{i_2}+a_{i_3}=B$. We set $S_j=\{a_{i_1},a_{i_2},a_{i_3}\}$. It remains to observe that $S_1,\ldots,S_m$ is a $3$-partition of $A$.    
This completes the proof of (v).

\begin{figure}
\begin{centering}
\includegraphics[scale=0.85]{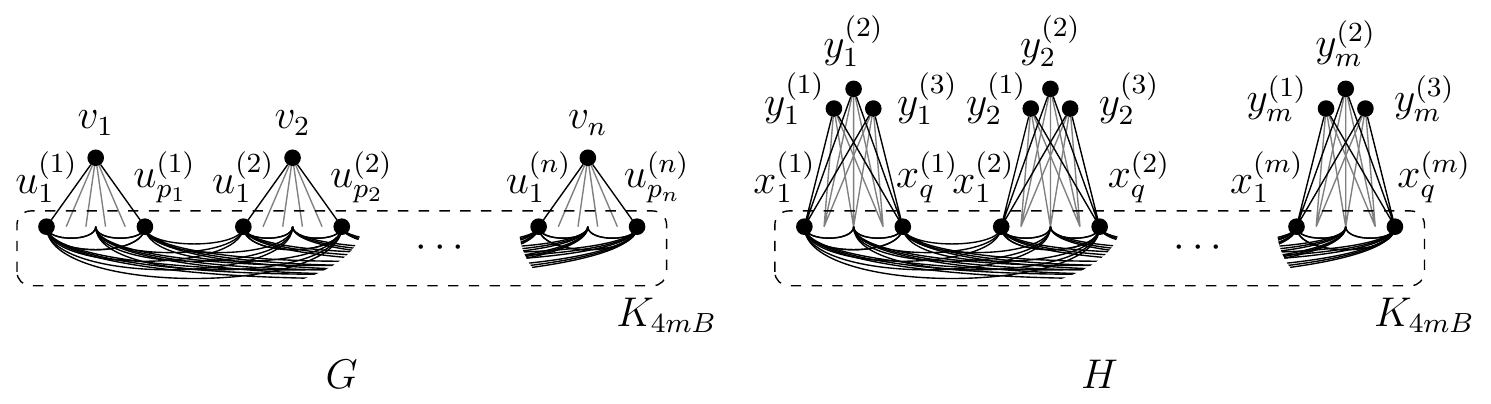}
\caption{The split graphs $G$ and $H$ constructed in the proof of (vi).}
\label{f-split}
\end{centering}
\end{figure}

\medskip
\noindent
{\bf (vi)} For $i=1,\ldots,n$, let $p_i=a_i+B$, and let $q=4B$.
We construct two graphs $G$ and $H$. We first construct $G$:
\begin{itemize} 
\item[$\bullet$] for  $i=1,\ldots,n$, introduce $p_i$ vertices $u_1^{(i)},\ldots,u_{p_i}^{(i)}$ and a vertex $v_i$ adjacent to $u_1^{(i)},\ldots,u_{p_i}^{(i)}$;
\item[$\bullet$] joint all $u$-vertices by edges pairwise to obtain a clique of size $4mB$.
\end{itemize}
We construct $H$ as follows:
\begin{itemize} 
\item[$\bullet$] for $j=1,\ldots,m$, introduce $q$ vertices $x_1^{(j)},\ldots,x_{q}^{(j)}$ and vertices $y_j^{(1)},y_j^{(2)},y_j^{(3)}$ adjacent to $x_1^{(j)},\ldots,x_{q}^{(j)}$;
\item[$\bullet$] joint all $x$-vertices by edges pairwise to obtain a clique of size $4mB$.
\end{itemize}
We observe that $G$ and $H$ are split graphs, also see Figure~\ref{f-split}.
We claim that $A$ has a $3$-partition if and only if there is a surjective homomorphism from $G$ to $H$.

First suppose that $S_1,\ldots,S_m$ is a $3$-partition of $A$. We define $f$ as follows. 
For $j=1,\ldots,m$, we consider the set $S_j=\{a_{i_1},a_{i_2},a_{i_3}\}$. 
We let $f$ map the vertices $v_{i_1},v_{i_2},v_{i_3}$ to $y_j^{(1)},y_j^{(2)},y_j^{(3)}$ respectively.
Then we let $f$ map the vertices $u_1^{(i_1)},\ldots,u_{p_{i_1}}^{(i_1)}$  to the first $p_{i_1}$ vertices of the set
$\{x_1^{(j)},\ldots,x_{q}^{(j)}\}$, the vertices 
$u_1^{(i_2)},\ldots,u_{p_{i_2}}^{(i_2)}$ 
to the next $p_{i_2}$ vertices of 
this set, and finally, the vertices $u_1^{(i_3)},\ldots,u_{p_{i_3}}^{(i_3)}$ to 
 the last $p_{i_3}$ vertices of the set. Because $p_{i_1}+p_{i_2}+p_{i_3}=a_{i_1}+a_{i_2}+a_{i_3}+3B=4B=q$, we find that $f$
 is a surjective homomorphism from $G$ to $H$.
  
Now suppose that $f$ is a surjective homomorphism from $G$ to $H$. Observe that $|V_G|=|V_H|$. Hence, $f$ is a bijection.
The homomorphism $f$ maps any clique of $G$ to a clique of the same size in $H$. It follows that all $u$-vertices of $G$ are mapped to $x$-vertices of $H$, and all $v$-vertices of $G$ are mapped to $y$-vertices of $H$. For $j=1,\ldots,m$, 
let $v_{i_1},v_{i_2},v_{i_3}$ be the vertices mapped to $y_j^{(1)},y_j^{(2)},y_j^{(3)}$ respectively. Then the vertices 
$u_1^{(i_1)},\ldots,u_{p_{i_1}}^{(i_1)}$, $u_1^{(i_2)},\ldots,u_{p_{i_2}}^{(i_2)}$ and $u_1^{(i_3)},\ldots,u_{p_{i_3}}^{(i_3)}$
are mapped bijectively to the vertices $x_1^{(j)},\ldots,x_{q}^{(j)}$. Therefore, $a_1+a_2+a_3+3B=p_1+p_2+p_3=q=4B$ and $a_1+a_2+a_3=B$.  
We set $S_j=\{a_{i_1},a_{i_2},a_{i_3}\}$, and it remains to observe that $S_1,\ldots,S_m$ is a $3$-partition of $A$.    
This completes the proof of (vi).

\begin{figure}
\begin{centering}
\includegraphics[scale=0.85]{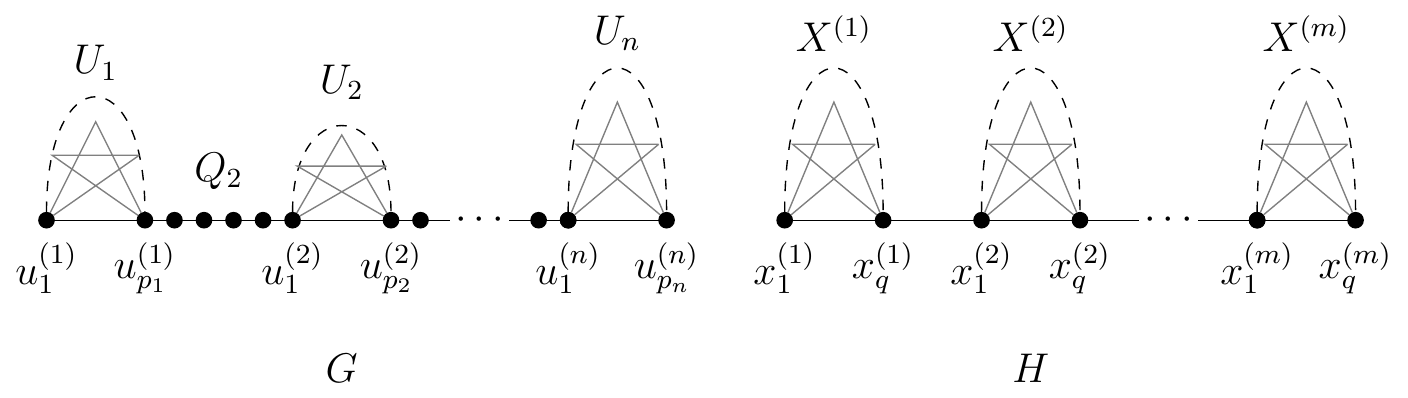}
\caption{The proper interval graphs $G$ and $H$ constructed in the proof of (vii).}
\label{f-pi}
\end{centering}
\end{figure}

\medskip
\noindent
{\bf (vii)} For $i=1,\ldots,n$, let $p_i=6m^2(a_i+B)$, and let $q=24m^2B$. 
We construct two graphs $G$ and $H$. We first construct $G$: 
\begin{itemize} 
\item[$\bullet$] for $i=1,\ldots,n$, construct a clique $U_i$ on $p_i$ vertices $u_1^{(i)},\ldots,u_{p_i}^{(i)}$; 
\item[$\bullet$] for $i=2,\ldots,n$, join $u_{p_{i-1}}^{(i-1)}$ and $u_{1}^{(i)}$
by a path $Q_i$ of length $2m-1$.
\end{itemize}
Then we construct $H$:
\begin{itemize} 
\item[$\bullet$] for $j=1,\ldots,m$, construct a clique $X^{(j)}$ on $q$ vertices $x_1^{(j)},\ldots,x_{q}^{(j)}$; 
\item[$\bullet$] for $j=2,\ldots,m$, join $x_{q}^{(j-1)}$ and $x_{1}^{(j)}$ by an edge.
\end{itemize}
We observe that $G$ and $H$ are proper interval graphs, also see Figure~\ref{f-pi}.
We claim that $A$ has a $3$-partition if and only if there exists a surjective homomorphism from $G$ to $H$.

First suppose that $S_1,\ldots,S_m$ is a $3$-partition of $A$.
We partition each $X^{(j)}$ into three cliques $X_{i_1}\cup X_{i_2}\cup X_{i_3}$ of size 
$p_{i_1}$, $p_{i_2}$, and $p_{i_3}$, respectively, corresponding to 
$S_j=\{a_{i_1},a_{i_2},a_{i_3}\}$; this is possible because $|X^{(j)}|=q=24m^2B=
6m^2(a_{i_1}+a_{i_2}+a_{i_3}+3B)=p_{i_1}+p_{i_2}+p_{i_3}$;
We will determine a homomorphism $f$ from $G$ to $H$ such that
$f$ is a bijection from $U_i$ to $X_i$ for $i=1,\ldots,n$. 
Hence, this property will ensure that $f$ is surjective. 
In order to do this, we must show that we do not violate the definition of a homomorphism with respect to the remaining vertices of $G$; note these remaining vertices are the  inner vertices of the $Q$-paths. We therefore define $f$ inductively as follows. 

 Let $i=1$. Assume that $a_1\in S_j$. We let $f$ map the vertices of $U_1$ to the vertices of $X_1$ bijectively in an arbitrary order. 
 
Let $i\geq 2$ and suppose that $f$ is constructed for all vertices of $U_s$ and $Q_s$ 
for all $1\leq s\leq i-1$. 
Let $y=f(u_{p_{i-1}}^{(i-1)})$.
Because $H$ has diameter at most $2m-1$, we find
that $y$ is at distance at most $2m-1$ from the set $X_i$.
Consider the subgraph $H'$ of $H$ that contains  $X_i$ and a shortest path between $y$ and $X_i$. Because $|X_i|\geq 2m$, we find that $H'$ contains a $(y,z)$-path of length $2m-1$ for some vertex $z\in X_i$. Recall that $y=f(u_{p_{i-1}}^{(i-1)})$. 
We map consecutively the vertices of the $(u_{p_{i-1}}^{(i-1)},u_{1}^{(i)})$-path $Q_i$ of length $2m-1$ to the vertices of $P$ in the path order. Note that $f(u_{1}^{(i)})=z$. Then we map the vertices $u_2^{(i)},\ldots,u_{p_{i}}^{(i)}$ to the vertices of $X_i\setminus\{z\}$ 
bijectively and in an arbitrary order. 
In this way we ensure that $f$ is a surjective homomorphism from $G$ to $H$. 

Now suppose that $f\colon V_G\rightarrow V_H$ is a surjective homomorphism. Because $f$ is a homomorphism, $f$ maps injectively every clique of $G$ to a clique in $H$. 
Because $p_i\geq 3$ for all $1\leq i\leq n$, we then find that $f$ cannot map a clique $U_i$ to an edge $x_q^{(j-1)}x_1^{(j)}$. Hence, $f$ maps $U_i$
injectively to some clique $X^{(j)}$ of $H$. 

Let $1\leq j\leq m$, and
let $\{i_1,\ldots,i_s\}$ be the set of all indices that correspond to the $U$-cliques that $f$ maps to $X_j$. Suppose that $p_{i_1}+\ldots+p_{i_s}<q$. Then, 
$6m^2(a_{i_1}+\ldots+a_{i_s}+sB)=p_{i_1}+\ldots+p_{i_s}<q=24m^2B$.
This means that $a_{i_1}+\ldots+a_{i_s}+sB\leq 3$. 
Consequently,
$q-(p_{i_1}+\ldots+p_{i_s})\geq 6m^2$. 
Hence, $f$ maps at least $6m^2$ inner vertices of the paths $Q_i$ to $X^{(j)}$. 
However, the total number of these vertices is $(n-1)(2m-2)=(3m-1)(2m-2)<6m^2$, a contradiction. This means that $p_{i_1}+\ldots+p_{i_s}\geq q$. Because the same claim holds for all $1\leq j\leq m$, and 
$p_1+\cdots +p_n=6m^2(a_1+\ldots + a_n+nB)=6m^2(mB+3mB)=24m^3B=mq$,
we conclude that $p_{i_1}+\ldots+p_{i_s}=q$. Because 
$6m^2(a_{i_1}+\ldots+a_{i_s}+sB)=p_{i_1}+\ldots+p_{i_s}=q=24m^2B$ and 
$a_{i_1}+\ldots+a_{i_s}>0$, we find that $s\leq 3$.
Then, because the same claim holds for all $1\leq j \leq m$, and 
$p_1+\cdots +p_n=6m^2(a_1+\ldots + a_n+nB)=mq=24m^3B$,
we find that $s=3$  and $a_{i_1}+a_{i_2}+a_{i_3}=B$.
We set $S_j=\{a_{i_1},a_{i_2},a_{i_3}\}$. It remains to observe that $S_1,\ldots,S_m$ is a $3$-partition of $A$.  This completes the proof of~(vii).
\qed
\end{proof}

\section{Tractable Cases}\label{sec:FPT}

By Theorem~\ref{t-all} (v), {\sc Surjective Homomorphism} is \NP-complete when $G$ and $H$ are restricted to be trees. Here, we prove that the problem is \FPT\  for trees when parameterized by the number of leaves in $H$. 
We first need some additional terminology. Let $T$ be a tree. Then we may fix some vertex of $T$ and call it the {\it root} of $T$.  
We observe that the root defines a parent-child relation between adjacent vertices. 
This enables us to define for a vertex $u$ of $T$ the tree $G_u$, which is the subtree of $T$ that is induced by
$u$ and all its descendants in $T$; we fix $u$ to be the root of $G_u$.
For a child $v$ of $u$, we let $G_{uv}$ denote the subtree of $G$ induced by $u$ and the set of all descendants of $v$ in $T$; we fix $u$ to be the root of $G_{uv}$.

\begin{theorem}\label{thm:trees-FPT}
Testing if there is a surjective homomorphism from an $n$-vertex tree $G$ to an $m$-vertex tree $H$ with $k$ leaves can be done in  $O(2^{2k}nm^2)$ time.
\end{theorem}

\begin{proof}
We use dynamic programming.
If $H$ has one vertex the claim is trivial.
Assume that $H$ has at least one edge.  Let $L$ be the set of the leaves of $H$.
First, we fix a root $r$ of $G$.
For each vertex $u\in V_G$, we construct a table that contains a number of  
records $R=(x,S)$ where $x\in V_H$ and $S\subseteq L$. 
A pair $(x,S)$ is a record for $u$ if and only if there exists a homomorphism $h$ from
$G_u$ to $H$ such that $h(u)=x$ and $S\subseteq h(V_{G_u})$. 
We also construct a similar table for each edge $uv\in E_G$. 
Then a pair $(x,S)$ is a record for $uv$ if and only if there exists a homomorphism $h$ from
$G_{uv}$ to $H$ such that $h(u)=x$ and $S\subseteq h(V_{G_{uv}})$. 
The key observation is that a homomorphism $f$ from $G$ to $H$ is surjective if and only
if $L\subseteq f(V_G)$, i.e., if and only if  the table for $r$ contains at least one record $(z,L)$.

We construct the tables as follows. We start with the leaves in $G$ not equal to $r$
(should $r$ be a leaf). Their tables are constructed straightforwardly.
Suppose that we have not constructed the table for a vertex $u$, while we have constructed the tables
for all children $v_1,\ldots,v_p$ of $u$. Then we first determine the table for each edge $uv_i$ by letting it consist of all records $(x,S)$ such that
\begin{itemize}     
\item[$\bullet$] $(y,S)$ with $y\in N_H(x)$ is in the table for $v_i$;
\item[$\bullet$] $x\in L$ and $(y,S\setminus\{x\})$ with $y\in N_H(x)$ is in the table for $v_i$. 
\end{itemize}
To construct the table for $u$, we consecutively construct auxiliary tables for $i=1,\ldots,p$. 
The table for $i=1$ is the table for $uv_1$.
The table for $i\geq 2$ 
consists of the records
$(x,S)$ such that $S=S'\cup S''$, $(x,S')$ is in the table for $i-1$  
and $(x,S'')$ is in the table for $uv_i$. The table for $u$ is the table constructed for $i=p$.

The correctness of the algorithm follows from its description. We observe that each table contains at most $m2^k$ records and can be constructed in  $O(2^{2k}\cdot m^2)$ time.
Because we construct $O(n)$ tables (including the auxiliary ones), our algorithms runs in $O(2^{2k}\cdot nm^2)$ time. This completes the proof of Theorem~\ref{thm:trees-FPT}.
\qed  
\end{proof}

We now prove that {\sc Surjective Homomorphism} is \FPT\ when parameterized by the vertex cover number of $G$ and $H$. 
The following approach has been successful before~\cite{FellowsLMRS08,FialaGK11}.
The idea is to reduce a problem to an integer linear programming
problem that is FPT when parameterized by the number of variables. 
Therefore, 
we consider the {\sc $p$-Variable Integer Linear Programming Feasibility} problem that has as input  
a $q\times p$ matrix $A$ with integer elements and an integer vector $b\in \mathbb{Z}^q$ and that is to decide whether there exists a vector $x\in\mathbb{Z}^p$ such that $A\cdot x\leq b$.
Lenstra~\cite{Lenstra83} showed that this problem is FPT when parameterized by $p$. 
The best running time is due to Frank and Tardos~\cite{FrankT87}.

\begin{lemma}[\cite{FrankT87}]\label{l:lin}
The {\sc $p$-Variable Integer Linear Programming Feasibility}
problem can be solved using $O(p^{2.5p+o(p)}\cdot L)$ arithmetic operations and space polynomial in $L$, where $L$ is the number of bits of the input.
\end{lemma}

\begin{figure}
\begin{centering}
\includegraphics[scale=0.9]{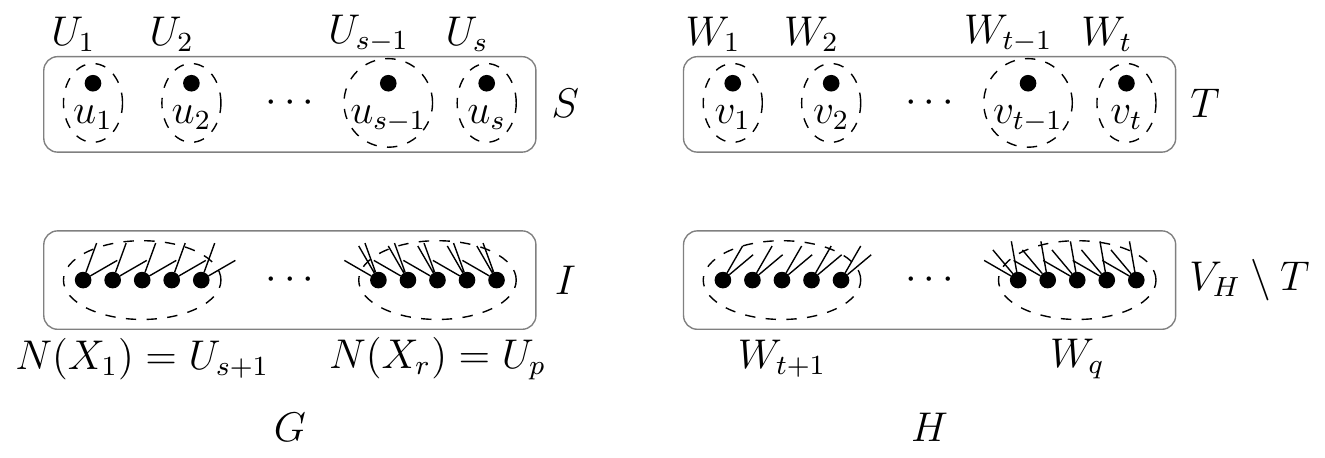}
\caption{The graphs $G$ and $H$ as considered in the proof of Theorem~\ref{thm:vcn-FPT}.}
\label{f-vc}
\end{centering}
\end{figure}

\begin{theorem}\label{thm:vcn-FPT}
Testing if there is a surjective homomorphism from an
$n$-vertex graph $G$ with $\vc(G)\leq k$ to an $m$-vertex graph $H$ with $\vc(H)\leq k$
can be done in $2^{2^{O(k)}}(nm)^{O(1)}$ time.
\end{theorem}

\begin{proof}
Let $G$ be an $n$-vertex graph with 
a vertex cover $S=\{u_1,\ldots,u_s\}$ of size $s\leq k$.
Then  $I = V_G\setminus S$ is an independent set. 
For every subset $X\subseteq S$, we define $N(X)$ as the set of vertices in $I$ that all have
neighborhood $X$, i.e., $N(X)=\{u\in I\; |\; N(u)=X\}$.
Note that $N(\emptyset)$ is the set of isolated vertices in $I$.

Let $X_1,\ldots,X_r\subseteq S$ be the sets with $N(X_i)\neq\emptyset$.
We let $p=s+r$ and define sets $U_1,\ldots,U_p$ where $U_i=\{u_i\}$ for $i=1,\ldots,s$ and 
$U_i=N(X_{i-s})$ for $i=s+1,\ldots,p$. 
We observe that $p\leq k+2^k$ and that 
$U_1,\ldots,U_p$ is a partition of $V_G$, where 
each $U_i$ is an independent set. 
Moreover, a vertex $v\in U_i$ is adjacent to a vertex $w\in U_j$ if and only if each vertex of $U_i$ is adjacent to each vertex of $U_j$. In that case, we say
that $U_i$ is \emph{adjacent} to $U_j$. We display $G$ in Figure~\ref{f-vc}.

Let $H$ be an $m$-vertex graph with a vertex cover $T=\{v_1,\ldots,v_t\}$ of size $t\leq k$. 
Then $J=V_H\setminus T$ is an independent set, and for each $Y\subseteq T$ we define
$N(Y)=\{z\in J\; |\; N(z)=Y\}$.
Then we define $q\leq k+2^k$ sets $W_1,\ldots,W_q$ where  $W_j=\{v_j\}$ for $j=1,\ldots,t$ and
$W_j=N(Y_{j-t})$ for $j=t+1,\ldots,q$.
We also display $H$ in Figure~\ref{f-vc}. The observations that we made for the $U$-sets are also valid for the $W$-sets.

Now we introduce integer variables $x_{ij}$ for $1\leq i\leq p$ and $1\leq j\leq q$, and observe that there is a surjective mapping (not necessarily a homomorphism) $f\colon V_G\rightarrow V_H$ such that $x_{ij}$ vertices 
of $U_i$ are mapped to $W_j$
if and only if the $x_{ij}$-variables satisfy the system
$$
\left\{
    \begin{array}{rcllllllll}
        x_{ij}&\geq&& 0 &&&&&& i\in\{1,\dots,p\},~j\in\{1,\dots,q\}\\[3pt]
        \sum_{j=1}^qx_{ij}&=&&|U_i|  &&&&&& i\in\{1,\dots,p\}\\ [3pt]
        \sum_{i=1}^px_{ij}&\geq&&|W_j| &&&&&& j\in\{1,\dots,q\}.\\ 
    \end{array}
\right.\eqno{(1)}$$
The mapping $f$ is a homomorphism from $G$ to $H$ if and only if the following holds: for each pair of variables $x_{ij},x_{i'j'}$ 
such that $x_{ij}>0$ and $x_{i'j'}>0$, if $U_i$ is adjacent to $U_{i'}$, then $W_j$ is adjacent to $W_{j'}$. 

We are now ready to give our algorithm.
We first determine the set $S$ and $T$. We then determine the $U$-sets and the $W$-sets.
We guess a set $R$ of indices $(i,j)$ and only allow the variables $x_{ij}$ for $(i,j)\in R$ to get non-zero value. Hence, we set $x_{ij}=0$ for $(i,j)\notin R$. We then check whether for all pairs $(i,j),(i',j')\in R$, if $U_i$ is adjacent to $U_{i'}$, then $W_j$ is adjacent to $W_{j'}$. If not, then we discard $R$ and guess a next one. Else we solve the system (1). If the system has an integer solution, then the algorithm returns {\sc Yes};
otherwise we try a next guess of $R$. If all guesses fail, then the algorithm returns {\sc No}.    

The correctness of the above algorithm follows from the aforementioned observations.
We now estimate the running time. We can find $S$ and $T$ in time 
$1.2738^kn^{O(1)}$ and $1.2738^km^{O(1)}$, respectively~\cite{ChenKX06}.
Then the sets $U_1,\ldots,U_p$ and $W_1,\ldots,W_q$ can be constructed in time 
$1.2738^k(nm)^{O(1)}$. The number of variables $x_{ij}$ is $pq\leq (k+2^k)^2=2^{O(k)}$. This means that  there are at most $2^{2^{O(k)}}$ possibilities to choose $R$. 
By Lemma~\ref{l:lin},  system (1) (with some variables $x_{ij}$ set to be zero) can be solved in time $2^{2^{O(k)}}(nm)^{O(1)}$. Hence, the total running time is $2^{2^{O(k)}}(nm)^{O(1)}$.
This completes the proof of Theorem~\ref{thm:vcn-FPT}.
\qed
\end{proof}

\section{Conclusions}\label{s-con}

Our complexity study shows that the {\sc Surjective Homomorphism} problem is already \NP-complete on a number of very elementary graph classes such as linear forests,
trees of small pathwidth, unions of complete graphs, cographs, split graphs and proper interval graphs.
We conclude that there is not much hope for finding tractable results in this direction, and consider the computational complexity classification of the {\sc Surjective $(-,H)$-Homomorphism} problem as the main open problem; note that  {\sc Surjective $(G,-)$-Homomorphism} is trivially polynomial-time solvable for any guest graph $G$.

As we observed in Section~\ref{s-intro}, the {\sc Surjective $(-,H)$-Homomorphism} problem is \NP-complete already for any fixed host graph $H$ that is nonbipartite. We also mentioned the existence of
a bipartite graph $H$ for which the problem is \NP-complete~\cite{BKM11} and that
the problem can be solved in polynomial time whenever the host graph $H$ is a fixed tree~\cite{GPS11}. The paper of Feder et al.~\cite{FHJKN10} on retractions provides a good starting point for the next step as we explain below.

A {\it pseudoforest} is a graph in which each connected component has at most one cycle. 
The {\sc Retraction} problem is to test whether a graph $G$ retracts to a graph $H$.
Feder et al.~\cite{FHJKN10} consider this problem for graphs that may have self-loops.
Applying their result to simple graphs yields the following. 
For any pseudoforest $H$, the $(-,H)$-{\sc Retraction} problem is  
\NP-complete if $H$ is nonbipartite or contains a cycle on at least $6$ vertices, and it is polynomial-time solvable otherwise.
It is an interesting open problem to show that $(-,H)$-{\sc Retraction} and 
{\sc Surjective $(-,H)$-Homomorphism} are polynomially equivalent for any fixed host graph $H$.
All the evidence so far seems to suggest this.

\end{document}